\documentclass[12pt]{article}
\usepackage{amsthm,amsmath,indentfirst,amssymb}
\usepackage{algorithm,algpseudocode,color}
\newtheorem{lemma}{Lemma}[section]
\newtheorem{theorem}[lemma]{Theorem}
\newtheorem{corollary}[lemma]{Corollary}
\newtheorem{definition}[lemma]{Definition}

\newtheorem{example}[lemma]{Example}

\begin{document}

\title{Approximation Algorithm for the Partial Set Multi-Cover Problem}
\author{\footnotesize Yishuo Shi$^{1}$\quad Yingli Ran$^2$\quad Zhao Zhang$^2$\quad James Willson$^3$\quad Guangmo Tong$^3$\quad Ding-Zhu Du$^3$\\
    {\it\small $^1$ College of Mathematics and System Sciences, Xinjiang University}\\
    {\it\small Urumqi, Xinjiang, 830046, China}\\
    {\it\small $^2$ College of Mathematics and Computer Science, Zhejiang Normal University}\\
    {\it\small Jinhua, Zhejiang, 321004, China}\\
{\it\small $^3$ Department of Computer Science, University of Texas at Dallas}\\
    {\it\small Richardson, Texas, 75080, USA}}

\date{}
\maketitle

\begin{abstract}
Partial set cover problem and set multi-cover problem are two generalizations of set cover problem. In this paper, we consider the partial set multi-cover problem which is a combination of them: given an element set $E$, a collection of sets $\mathcal S\subseteq 2^E$, a total covering ratio $q$ which is a constant between 0 and 1, each set $S\in\mathcal S$ is associated with a cost $c_S$, each element $e\in E$ is associated with a covering requirement $r_e$, the goal is to find a minimum cost sub-collection $\mathcal S'\subseteq\mathcal S$ to fully cover at least $q|E|$ elements, where element $e$ is fully covered if it belongs to at least $r_e$ sets of $\mathcal S'$.
Denote by $r_{\max}=\max\{r_e\colon e\in E\}$ the maximum covering requirement. We present an $(O(\frac{r_{\max}\log^2n}{\varepsilon}),1-\varepsilon)$-bicriteria approximation algorithm, that is, the output of our algorithm has cost at most $O(\frac{r_{\max}\log^2 n}{\varepsilon})$ times of the optimal value while the number of fully covered elements is at least $(1-\varepsilon)q|E|$.

\end{abstract}

\section{Introduction}
 Set cover is an extensively studied problem in combinatorial optimization. In this paper, we study a variant of the set cover problem, namely the {\em partial set multi-cover problem}, which is defined as follows.

\begin{definition}[Partial Set Multi-Cover (PSMC)]\label{defPMSPC}
{\rm  Suppose $E$ is an element set, $\mathcal{S}\subseteq 2^{E}$ is a collection of subsets of $E$, each set $S\in \mathcal{S}$ has a cost $c_S$, each element $e\in E$ has a positive covering requirement $r_e$. For a sub-collection $\mathcal S'\subseteq \mathcal S$, denote by $\mathcal S'_e=\{S\in\mathcal S'\colon e\in S\}$ those sets of $\mathcal S'$ containing element $e$. If $|\mathcal S'_e|\geq r_e$, we say that $e$ is fully covered by $\mathcal S'$, denoted as $e\sim \mathcal S'$. The cost of sub-collection $\mathcal S'$ is $c(\mathcal S')=\sum_{S\in \mathcal S'}c(S)$. Given $E,\mathcal S,c,r$ with $|E|=n$ and a real number $q$ which is a constant between 0 and 1, the PSMC problem is to find a minimum cost sub-collection $\mathcal S'$ such that $|\{e\in E\colon e\sim\mathcal S'\}|\geq qn$. An instance of PSMC is denoted as $(E,\mathcal S,c,r,q)$. }
\end{definition}

The PSMC problem includes two important variants of the set cover problem. When $r_e\equiv 1$, it is the \emph{partial set cover problem} (PSC). When $q=1$, it is the \emph{set multi-cover problem} (SMC). One motivation of PSC comes from the phenomenon that in a real world, ``satisfying all requirements'' will be too costly or even impossible, due to resource limitation or political policy \cite{Charikar}. And SMC comes from the requirement of fault tolerance in practice \cite{Zhang}. There are a lot of researches on PSC and SMC, achieving performance ratios matching the lower bounds for the classic set cover problem, namely $\ln n$ and $f$, where $n$ is the number of elements and $f$ is the maximum number of sets containing a common element. However, study on the combination of these two problems is very rare. According to our recent paper \cite{Ran3}, under the ETH assumption, the PSMC problem cannot be approximated within factor $O(n^\frac{1}{2(\log \log n)^c})$ for some constant $c$.

The aim of this paper is to explore a greedy strategy on PSMC.

\subsection{Related Work}\label{secRW}

The set cover problem (SC) was one of the first 21 problems proved to be NP-hard in Karp's seminal paper \cite{Karp}. In fact, Feige \cite{Feige} proved that it cannot be approximated within factor $(1-o(1))\ln n$ unless $NP\subseteq DTIME(n^{O(\log\log n)})$, where $n$ is the number of elements. Dinur and Steurer \cite{Dinur} proved the same lower bound under the assumption that $P\neq NP$. Khot and Regev \cite{Khot} showed that it cannot be approximated within factor $f-\varepsilon$ for any constant $\varepsilon>0$ assuming that unique games conjecture is true, where $f$ is the maximum number of sets containing a common element. On the other hand, greedy strategy achieves performance ratio $H(\Delta)\leq \ln\Delta+1$ \cite{Chvatal,Johnson,Lovasz}, where $\Delta$ is the maximum cardinality of a set and $H(\Delta)=1+\frac{1}{2}+\cdots+\frac{1}{\Delta}$ is the Harmonic number. And $f$-approximation exists by either LP rounding method \cite{Hochbaum} or local ratio method \cite{BarYuhuda}.

In paper \cite{Dobson}, Dobson first gave an $H(K)$-approximation algorithm for multi-set multi-cover problem (MSMC), where $K$ is the maximum size of a multi-set. Rajagopalan and Vazirani \cite{Rajagopalan} gave a greedy algorithm achieving the same performance ratio, using dual fitting analysis, which implies that the integrality gap of the classic linear program of MSMC is at most $H(K)$.

For the partial set cover problem, Kearns \cite{Kearns} gave a greedy algorithm achieving performance ratio $2H(n)+3$. By modifying the greedy algorithm a little, Slavik \cite{Slavik} improved the performance ratio to $H(\min \{\lceil qn\rceil,\Delta\})$, where $q$ is the percentage that elements are required to be covered. Gandhi {\it et al.} \cite{Gandhi} proposed a primal-dual algorithm achieving performance ratio $f$. Bar-Yuhuda \cite{Bar-Yuhuda} studied a generalized version in which each element has a profit and the total profit of covered elements should exceed a threshold. Using local ratio method, he also obtained performance ratio $f$. Konemann {\it et al.} \cite{Konemann} presented a Lagrangian relaxation framework and obtained performance ratio $(\frac{4}{3}+\varepsilon)H(\Delta)$ for the generalized partial set cover problem.

From the above related work, it can be seen that both PSC and SMC admit performance ratios matching those best ratios for the classic set cover problem. However, combining partial set cover with set multi-cover has enormously increased the difficulty of studies. Ran \emph{et al.} \cite{Ran} were the first to study approximation algorithms for PSMC, using greedy strategy and dual-fitting analysis. However, their ratio is meaningful only when the covering percentage $q$ is very close to $1$. In paper \cite{Ran2}, the authors presented a simple greedy algorithm achieving performance ratio $\Delta$. They also presented a local ratio algorithm, which reveals a what they called ``shock wave'' phenomenon: their performance ratio is $f$ for both PSC and SMC , however, when $q$ is smaller than 1 by a very small constant, the ratio jumps abruptly to $O(n)$. In our recent paper \cite{Ran3}, we proved that PSMC cannot have a better than polynomial performance ratio by a reduction from the well-known {\em densest $k$-subgraph problem}.

\subsection{Our Contribution and Techniques}\label{subsecContribution}

The contributions of this paper is summarized as follows.
\begin{itemize}
\item A new problem called {\em minimum density sub-collection} (MDSC) is defined, which is to find a sub-collection $\mathcal S'\subseteq \mathcal S$ to minimize the ratio $c(\mathcal S')/|\mathcal C(\mathcal S')|$, where $\mathcal C(\mathcal S')$ is the set of elements fully covered by $\mathcal S'$. We prove that MDSC is also NP-hard.
\item We show that if MDSC has an $\alpha$-approximation algorithm, then PSMC has an $(O(\frac{\alpha}{\varepsilon}),1-\varepsilon)$-bicriteria approximation algorithm, that is, the output of our algorithm has cost at most $O(\frac{\alpha}{\varepsilon})$ times that of an optimal solution, while the total number of fully covered elements is at least $q(1-\varepsilon)n$, where $q$ is the covering ratio required by the problem and $n$ is the total number of elements.
\item We design an $O(r_{\max}\log^2n)$-approximation algorithm for MDSC, where $r_{\max}=\max_{e\in E}r_e$ is the maximum covering requirement of elements. Combining this result with the above, PSMC has an $(O(\frac{r_{\max}\log^2n}{\varepsilon}),1-\varepsilon)$-bicriteria algorithm.
\end{itemize}

Our algorithm uses a greedy strategy. However, there is a problem of which sets should be chosen in each iteration. As indicated by previous studies in paper \cite{Ran}, a natural generalization of the classic greedy algorithm cannot yield good results. One reason might be that the number of elements fully covered by a sub-collection of sets is {\em not} submodular. In this paper, our greedy algorithm iteratively picks an approximate solution to the MDSC problem until the number of elements which are fully covered reaches a certain degree. An obstacle to obtaining a good approximation factor lies in the last iteration: the sub-collection chosen in the last iteration might cover much more elements than required.  Although its density is low, its cost might be too large to be bounded by the optimal value. So, we stop the algorithm when at least $q(1-\varepsilon)n$ elements are fully covered, and thus leading to a bicriteria approximation algorithm.

A crucial stepping-stone to the above algorithm is the MDSC problem, which is also NP-hard. An example can be constructed showing that a natural LP formulation has integrality gap arbitrarily large. To overcome such a difficulty, we formulate the problem as a linear program using a language having a taste of ``flow'' and made use of an approximation algorithm for the minimum node-weighted Steiner network problem as a subroutine to yield a performance guaranteed approximation algorithm for MDSC.

The paper is organized as follows. In Section \ref{secPre}, we give the definition of MDSC and prove its NP-hardness. In Section \ref{secPSMC}, we show how an $\alpha$-approximation algorithm for MDSC leads to an $(O(\frac{\alpha}{\varepsilon}),1-\varepsilon)$-bicriteria algorithm for PSMC. In Section \ref{secMDSC}, we propose an $O(r_{\max}\log^2 n)$-approximation algorithm for MDSC. In Section \ref{secConclude}, the paper is concluded with some discussions on future work.

A preliminary version of this paper was presented in INFOCOM2017. There is a flaw in that version. We explain in the appendix where is the flaw.

\section{Preliminaries}\label{secPre}
For simplicity of statement, we shall use $\mathcal C(\mathcal S')$ to denote the set of elements fully covered by sub-collection $\mathcal S'$. Define the {\em density} of a sub-collection $\mathcal S'$ as
$$den(\mathcal S')=c(\mathcal S')/|\mathcal C(\mathcal S')|.$$

\begin{definition}[Minimum Density Sub-Collection (MDSC)]
{\rm Given $E,\mathcal S,c,r$, the MDSC problem is to find a sub-collection with the minimum density}.
\end{definition}

Unfortunately, MDSC is also NP-hard.

\begin{theorem}
The MDSC problem is NP-hard.
\end{theorem}
\begin{proof}
We reduce the perfect 3-dimensional matching problem (which is APX-hard \cite{Ausiello}) to MDSC. Given an integer $k$, three sets $X,Y,Z$ each having cardinality $k$, and a set $T\subseteq X\times Y\times Z$, the perfect 3-dimensional matching problem asks whether there is a subset $T'\subseteq T$ with $|T'|=k$ such that for any elements $(x,y,z),(x',y',z')\in T'$, $x\neq x'$, $y\neq y'$, and $z\neq z'$. Construct an instance of MDSC as follows. Let $E=X\cup Y\cup Z\cup \{u_0\}$ and $\mathcal S=\{\{x,y,z,u_0\}\colon (x,y,z)\in T\}$. The covering requirement $r_{u_0}=k$ and $r_u=1$ for $u\in X\cup Y\cup Z$. The cost $c_S=1$ for all $S\in\mathcal S$.

Next, we show that there is a perfect 3-dimensional matching if and only if the optimal value for the MDSC problem is $k/(3k+1)$. In fact, if $T'$ is a perfect 3-dimensional matching, then $\mathcal S'=\{\{x,y,z,u_0\}\colon (x,y,z)\in T'\}$ has $|\mathcal S'|=|T'|=k$ and $|\mathcal C(\mathcal S')|=3k+1$. Suppose the instance does not have a perfect 3-dimensional matching. Consider an arbitrary sub-collection $\mathcal S''$ and its corresponding subset $T''=\{(x,y,z)\colon \{x,y,z,u_0\}\in\mathcal S''\}$. Then $|T''|= |S''|=c(\mathcal S'')$. If $|T''|>k$, then $c(\mathcal S')/|\mathcal C(\mathcal S')|>k/(3k+1)$. If $|T''|<k$, then $u_0\not\in\mathcal C(\mathcal S'')$ and thus $|\mathcal C(\mathcal S'')|\leq 3|T''|$. In this case, $c(\mathcal S'')/|\mathcal C(\mathcal S'')|\geq |T''|/(3|T''|)>k/(3k+1)$. The claimed result is proved.
\end{proof}

\section{Bicriteria Algorithm for PSMC}\label{secPSMC}

In this section, we make use of an $\alpha$-approximation algorithm for MDSC to design a bicriteria algorithm for PSMC.

\vskip 0.2cm\subsection{The Algorithm}\label{secPMSMC}
The algorithm is presented in Algorithm \ref{main}. It follows the classic greedy strategy. A main difference is that instead of choosing sets one by one, in each iteration, it implements an $\alpha$-approximation algorithm for MDSC to greedily choose sub-collections. After each iteration, the instance is updated with respect to the current sub-collection $\mathcal F$ to form a reduced instance $(E',\mathcal S',c,r',q')$, where $E'=E-\mathcal C(\mathcal F)$ is the set of elements not having been fully covered, the total remaining covering ratio
\begin{equation}\label{eq1-25-2}
q'=\frac{qn-|\mathcal C(\mathcal F)|}{n},
\end{equation}
the remaining covering requirement for element $e$ is $r'_e=\max\{0,r_e-|\mathcal F_e|\}$, and those elements which have been fully covered by $\mathcal F$ have to be removed from each set. In the following, when we mention a {\em reduced instance} or when we say that the instance is {\em updated}, it is always understood that the above operations are executed. When the algorithm terminates, we have $q'\leq \varepsilon q$, and thus the number of fully covered elements is at least $(1-\varepsilon)qn$ by the expression of reduced covering ratio $q'$ defined in \eqref{eq1-25-2}.

\begin{algorithm}[ht!]
\caption{\textbf{{\sc Algorithm For PSMC via MDSC}}} \label{main}
{\bf Input:} A PSMC instance $(E,\mathcal S,c,r,q)$ and a real number $0<\varepsilon <1$.

{\bf Output:} A sub-collection $\mathcal F$ fully covering at least $(1-\varepsilon)qn$ elements.
\begin{algorithmic}[1]
\State $\mathcal F\leftarrow \emptyset$, $q'\leftarrow q$.
\While{$q'>\varepsilon q$}
    \State Use an $\alpha$-approximation algorithm for MDSC on the reduced instance to find a sub-collection $\mathcal R$.\label{algo-line3}
    \State $\mathcal F\leftarrow\mathcal F\cup\mathcal R$.
    \State Update the instance.
\EndWhile
\State Output $\mathcal F$.
\end{algorithmic}
\end{algorithm}

\vskip 0.2cm\subsection{Performance Ratio Analysis}\label{sec2}

Suppose Algorithm \ref{main} is executed $t$ times, selecting sub-collections $\mathcal R_1,\ldots,\mathcal R_t$. We estimate costs $\sum_{i=1}^{t-1}c(\mathcal R_i)$ and $c(\mathcal R_t)$ separately. In the following, $OPT$ denotes an optimal solution to PSMC, and $opt=c(OPT)$ is the optimal cost.

\begin{lemma}\label{lem16-3-11-1}
$\sum_{i=1}^{t-1}c(\mathcal R_i)\leq \alpha\ln\left(\frac{1}{\varepsilon}\right) \cdot opt.$
\end{lemma}
\begin{proof}
For $i=1,2,\ldots,t$, denote by $n_i$ the number of elements remaining to be fully covered after $\mathcal R_i$ is selected. Then $|\mathcal C(\mathcal R_i)|=n_{i-1}-n_i$ for $i=1,\ldots,t-1$ where

\begin{align}
& n_0=qn\ \mbox{and}\label{eq1-25-3}\\
& n_{t-1}=qn-\sum^{t-1}_{i=1}|\mathcal C(\mathcal R_{i})|>qn-(1-\varepsilon)qn=\varepsilon qn.\label{eq1-25-31}
\end{align}
After the $(i-1)$-th iteration, $OPT$ is a sub-collection fulfilling the remaining covering requirement $n_{i-1}$. So the density of an optimal solution $\mathcal R^*_i$ to the MDSC problem in the $i$-th iteration is upper bounded by $opt/n_{i-1}$. Since $\mathcal R_i$ approximates the density of $\mathcal R^*_i$ within a factor of $\alpha$, we have
\begin{equation}\label{eq1-26-1}
\frac{c(\mathcal R_i)}{n_{i-1}-n_i}\leq \alpha\frac{opt}{n_{i-1}}.
\end{equation}
Combining this with inequalities \eqref{eq1-25-3}, \eqref{eq1-25-31}
we have
\begin{align*}
 & \sum_{i=1}^{t-1}c(\mathcal R_i)
\leq \alpha\cdot opt \sum_{i=1}^{t-1}\frac{n_{i-1}-n_i}{n_{i-1}}\leq \alpha\cdot opt \sum_{i=1}^{t-1}\int_{n_i}^{n_{i-1}}\frac{1}{x}dx \\
= & \alpha\cdot opt \int_{n_{t-1}}^{n_0}\frac{1}{x}dx=\alpha \ln\left(\frac{n_{0}}{n_{t-1}}\right)opt \leq \alpha \ln\left(\frac{1}{\varepsilon}\right) opt.
\end{align*}
The lemma is proved.
\end{proof}

\begin{lemma}\label{lem16-3-11-2}
$c(\mathcal R_t)\leq \alpha\left(1+\frac{1-q}{\varepsilon q}\right) opt$.
\end{lemma}
\begin{proof}

For the last sub-collection $R_t$, notice that in a worst case, it may fully cover all the remaining elements, the number of which is $n-(qn-n_{t-1})=(1-q)n+n_{t-1}$. Compared with the density of the optimal sub-collection and similar to the derivation of \eqref{eq1-26-1},
$$
\frac{c(\mathcal R_t)}{(1-q)n+n_{t-1}}\leq \alpha\frac{opt}{n_{t-1}}.
$$
Combining this with inequality \eqref{eq1-25-31},
$$
c(\mathcal R_t) \leq \alpha\frac{(1-q)n+n_{t-1}}{n_{t-1}} opt< \alpha\left(1+\frac{1-q}{\varepsilon q}\right)opt.
$$
The lemma is proved.
\end{proof}

Since the output of Algorithm \ref{main} is $\mathcal R_1\cup\cdots\cup\mathcal R_{t-1}\cup \mathcal R_{t}$, the performance ratio follows from Lemma \ref{lem16-3-11-1} and Lemma \ref{lem16-3-11-2}.

\begin{theorem}\label{th1}
Implementing an $\alpha$-approximation algorithm for MDSC, the PSMC problem admits an $\left(\alpha\left(1+\ln(\frac{1}{\varepsilon})+\frac{1-q}{\varepsilon q}\right),1-\varepsilon\right)$-bicriteria approximation.
\end{theorem}

For small $\varepsilon$, the performance ratio in the above theorem can be viewed as $O(\frac{\alpha}{\varepsilon})$ since $q$ is a constant.

\section{Approximation Algorithm for MDSC}\label{secMDSC}

In this section, we present an approximation algorithm for MDSC. The algorithm is based on an LP formulation and makes use of a node weighted Steiner network algorithm.

\vskip 0.2cm\subsection{LP-Formulation}

The following is a natural formulation of integer program for MDSC.
\begin{align}\label{eq12-22-1}
\min\ \frac{\sum_{S\in\mathcal S}c_Sx_S}{\sum_{e\in E}y_e} &  \nonumber\\
 s.t.\    \sum_{S:\ e\in S}x_S  & \geq r_{e}y_e, \ \mbox{for any}\ e\in E\\
   x_S \in  \{0,1\} & \ \ \mbox{for}\ S\in\mathcal S\nonumber\\
   y_e \in  \{0,1\} & \ \ \mbox{for}\ e\in E\nonumber
\end{align}
Here $x_S$ indicates whether set $S$ is selected and $y_e$ indicates whether element $e$ is fully covered. The first constrained says that if $y_e=1$ then at least $r_e$ sets containing $e$ must be selected and thus $e$ is fully covered. Relaxing \eqref{eq12-22-1} and by a scaling, we have the following linear program:
\begin{align}\label{eq12-22-2}
\min\ \sum_{S\in\mathcal S}c_Sx_S &   \nonumber\\
\ \ s.t. \ \sum_{e\in E}y_e=1 & \nonumber \\
    \sum_{S:\ e\in S}x_S  & \geq r_{e}y_e, \ \mbox{for any}\ e\in E\\
  1\geq x_S& \geq 0\ \mbox{for}\ S\in\mathcal S\nonumber\\
  1\geq y_e& \geq 0\ \mbox{for}\ e\in E\nonumber
\end{align}
However, the following example shows that the integrality gap between \eqref{eq12-22-1} and \eqref{eq12-22-2} can be arbitrarily large.

\begin{example}\label{ex1}
{\rm Let $E=\{e_1,e_2\}$, $\mathcal S=\{S_1,S_2,S_3\}$ with $S_1=\{e_1\}$, $S_2=\{e_2\}$, $S_3=\{e_1,e_2\}$, $c(S_1)=c(S_2)=1$, $c(S_3)=M$ where $M$ is a large positive number, and $r(e_1)=r(e_2)=2$. Then $x_{S_1}=x_{S_2}=1$, $x_{S_3}=0$, $y_{e_1}=y_{e_2}=1/2$ form a feasible solution to \eqref{eq12-22-2} with objective value 2. While among all integral feasible solutions to \eqref{eq12-22-1}, which are $\{S_1,S_3\}$, $\{S_2,S_3\}$, $\{S_1,S_2,S_3\}$, whose densities are $M+1$, $M+1$, and $(2+M)/2$, respectively, the minimum density is $(2+M)/2$.}  \end{example}

Hence, to obtain a good approximation, we need to find another program. In the following, we formulate the problem in an more involved flow-like language. For an element $e$, an {\em $r_e$-cover-set} is a sub-collection $\mathcal S'$ with $|\mathcal S'|=r_e$ which fully covers $e$. Denote by $\Omega_e$ the family of all $r_e$-cover-sets, and $\Omega=\bigcup_{e\in E}\Omega_e$. Consider the following example.

\begin{example}\label{eg5-10-1}
{\rm $E=\{e_1,e_2,e_3\}$. $\mathcal S=\{S_1,S_2,S_3\}$ with $S_1=\{e_1,e_2,e_3\}$, $S_2=\{e_1\}$ and $S_3=\{e_1,e_3\}$, and $r(e_1)=2$ and $r(e_2)=r(e_3)=1$. For this example, $\Omega_{e_1}=\{\{S_1,S_2\},\{S_1,S_3\},\{S_2,S_3\}\}$, $\Omega_{e_2}=\{\{S_1\}\}$ and
$\Omega_{e_3}=\{\{S_1\},\{S_3\}\}$. It should be emphasized that a same cover-set belonging to different $\Omega_e$'s will be viewed as different cover-sets. For example, $\{S_1\}$ belongs to both $\Omega_{e_2}$ and $\Omega_{e_3}$. To distinguish them, we shall use $Q_j^{(e_i)}$ to denote cover-sets in $\Omega_{e_i}$. For example, $\Omega_{e_1}$ contains three $r_{e_1}$-cover-sets $Q_1^{(e_1)}=\{S_1,S_2\}$, $Q_2^{(e_1)}=\{S_1,S_3\}$ and $Q_3^{(e_1)}=\{S_2,S_3\}$, $\Omega_{e_2}$ contains one $r_{e_2}$-cove-set $Q_1^{(e_2)}=\{S_1\}$, $\Omega_{e_3}$ contains two $r_{e_3}$-cover-sets $Q_1^{(e_3)}=\{S_1\}$ and $Q_2^{(e_3)}=\{S_3\}$.}
\end{example}

The following is an integer program for constrained MDSC:
\begin{align}
\min &\ \frac{\sum_{S\in\mathcal S}c_Sx_S}{\sum_{e\in E}y_e}\label{eq11}\\
s.t. &\ \sum_{ Q: Q\in \Omega_e} l_Q\geq y_{e} \ \mbox{for}\ e\in E\nonumber\\
&\ x_S\geq \sum_{Q:S\in Q\in \Omega_e} l_Q \ \mbox{for}\ e\in E,S\in\mathcal S\nonumber\\
&\ x_S\in\{0,1\},\ \mbox{for}\ S\in\mathcal S\nonumber\\
&\ y_e\in\{0,1\},\ \mbox{for}\ e\in E\nonumber\\
&\ l_Q\in\{0,1\},\ \mbox{for every}\ Q\in \Omega_e\ \mbox{for some}\ e\in E\nonumber
\end{align}
In fact, $l_Q$ indicates whether a cover-set $Q$ is selected and $x_S$ indicates whether set $S$ is selected. The first constraint says that if $y_e=1$ then at least one $r_e$-cover-set is selected and thus $e$ is fully covered. The family of selected sets is the union of all those selected cover-sets. So, if $S$ belongs to some selected cover-set, then $x_S$ should be $1$, namely,
\begin{equation}\label{eq0902-1}
x_S\geq\max\{l_Q\colon S\in Q\in \Omega\}.
\end{equation}
Notice that to fully cover element $e$, it is sufficient to select exactly one $r_e$-cover-set from $\Omega_e$. So, we may replace \eqref{eq0902-1} by the second constraint of \eqref{eq11} for the purpose of linearization. The object function is exactly the density of selected sets.

Consider Example \ref{eg5-10-1} again. Setting $l_{Q_2^{(e_1)}}=l_{Q_1^{(e_3)}}=1$ and all other $l$-values to be 0 implies that the selected sub-collection $\mathcal S'=Q_2^{(e_1)}\cup Q_1^{(e_3)}=\{S_1,S_3\}$ and $e_1,e_3$ are fully covered. By the second constraint, $x_{S_1}=x_{S_3}=1$ and we may take $x_{S_2}=0$ (to minimize the objective function, it is better to take $x_S$ to be 0 if the right hand side of the second constraint is 0). By the first constraint, $y_{e_2}=0$ and we may take $y_{e_1}=y_{e_2}=1$ (to minimize the objective function, it is better to take $y_e$ to be 1 for all those elements $e$ which are fully covered). Notice that $\{S_1\}$ serves as both $Q_1^{(e_2)}$ and $Q_1^{(e_3)}$, the $l$-value for the former is 0 and the $l$-value for the latter is $1$, they are set independently.

The above integer program \eqref{eq11} can be relaxed to the following linear program LP$_1$:
\begin{align}
\min &\ \sum_{S\in\mathcal S}c_Sx_S\label{eq12}\\
s.t. &\ \sum_{e\in E}y_e=1\nonumber\\
&\ \sum_{ Q: Q\in \Omega_e} l_Q\geq y_{e} \ \mbox{for}\ e\in E\nonumber\\
&\ x_S\geq \sum_{Q:S\in Q\in \Omega_e} l_Q \ \mbox{for}\ e\in E, S\in\mathcal S\nonumber\\
&\ x_S\geq 0\ \mbox{for}\ S\in\mathcal S\nonumber\\
&\ y_e\geq 0\ \mbox{for}\ e\in E\nonumber\\
&\ l_Q\geq 0\ \mbox{for}\ Q\in \Omega_e \ \mbox{for some}\ e\in E.\nonumber
\end{align}

It should be noticed that although there is exponential number of variables, the linear program can be solved in polynomial time, the detail of which is presented as follows. Consider the dual program of \eqref{eq12}:
\begin{align}\label{eq11-22-6}
\max\ & a  \nonumber\\
s.t.\ & a-f_e\leq 0  \ \mbox{for}\ e\in E\\
 &    d_{S_e} \leq c_S,  \ \mbox{for}\ e\in S \in \mathcal S,\ e\in E \nonumber\\
 & f_e \leq \sum_{S: e\in S\in Q}d_{S_e}, \ \mbox{for} \ Q \in \Omega_{e}, \ e\in E \nonumber\\
 & f_e\geq 0\ \ \mbox{for}\ e\in E\nonumber\\
 & d_{S_e} \geq 0 \ \mbox{for}\ e \in E,\ S\in \mathcal S.\nonumber
\end{align}
By LP primal-dual theory \cite{Ignizio}, one may solve \eqref{eq12} through solving \eqref{eq11-22-6}, and to solve \eqref{eq11-22-6} in polynomial time, it suffices to construct a separation oracle for the third constraint. For any $e\in E$ and $Q \in \Omega_{e}$, define $g(e,Q)=\sum_{S: e\in S\in Q}d_{S_e}$. For any element $e\in E$, a cover-set $Q_{\min}(e)$ minimizing $g(e,Q)$ can be found by choosing the $r_e$ cheapest (measured by cost $d$) sets containing $e$. By checking whether $g(e,Q_{\min}(e))\geq f_e$ holds for every $e\in E$, we can either claim the validity of the constraints or find out a violated constraint. Using ellipsoid method, linear program \eqref{eq11-22-6} is polynomial-time solvable.

\begin{lemma}\label{lem16-3-10}
The optimal value of linear program \eqref{eq12}, denoted as $opt_{LP_1}$, satisfies $opt_{LP_1}\leq opt_{MDSC}$, where $opt_{MDSC}$ is the optimal value for integer linear program \eqref{eq11}.
\end{lemma}
\begin{proof} Let $(x^*,y^*,Q^*)$ be an optimal solution to \eqref{eq11}. Suppose $\sum_{e\in E}y_e^*=P^*$. Then $(x^*/P^*,y^*/P^*,Q^*/P^*)$ is a feasible solution to \eqref{eq12}. Hence $opt_{LP_1}\leq \sum_{S\in\mathcal S}c_S(x_S^*/P^*)=\sum_{S\in\mathcal S}c_Sx_S^*/\sum_{e\in E}y_e^*$.
\end{proof}

\vskip 0.2cm \subsection{The Algorithm}
Inspired by the method of paper \cite{Chekuri1} for network design problems, we design an approximation algorithm for MDSC which makes use of an approximation algorithm for the minimum node weighted Steiner network problem.

\begin{definition}[Node Weighted Steiner Network Problem (NWSN) \cite{Nutov}]
{\rm Given a graph $G=(V,E)$ with a weight function $c$ on $V$ and a connectivity requirement $r_{s,t}$ for each pair of nodes $(s,t)$, the minimum node weighted Steiner network problem asks for a subgraph $H$ such that every pair of nodes $(s,t)$ are connected by at least $r_{s,t}$ edge-disjoint paths in $H$ and the node weight of $H$ is as small as possible.}
\end{definition}

Notice that $H$ must include all those nodes $s$ with $r_{s,t}\neq 0$ for at least one node $t$. Such a node $s$ can be viewed as a {\em terminal node}. On the other hand, those nodes $s$ with $r_{s,t}=0$ for any $t\neq s$ need not be included in $H$. Such nodes are {\em Steiner nodes}. The NWSN problem is to select a set of Steiner nodes with the minimum weight to satisfy those connectivity requirements between terminal nodes.

The algorithm is presented in Algorithm \ref{algo5-11-1}. It partitions elements into disjoint union of sets $Y_i$'s, according to an optimal fractional solution to linear program \eqref{eq12}. Let $Y_{i_0}$ be a set satisfying the condition specified by line 3 of the algorithm, whose existence will be shown later. The NWSN instance used in line 4 of the algorithm is constructed in the following way. Let $H$ be the graph on node set $Y_{i_0}\cup \mathcal S\cup \{s\}$ and edge set $\{eS\colon e\in Y_{i_0},S\in\mathcal S,e\in S\}\cup\{sS\colon S\in \mathcal S\}$. Set the weight on every $S\in \mathcal S$ to be $c_S$ and the weight on all other nodes to be zero. Set the connectivity requirement $r_{s,e}=r_e$ for every $e\in Y_{i_0}$ and the connectivity requirement on all other node pairs to be zero. Denote the constructed instance as $(H,c,r)$.  The output of the algorithm is the sub-collection of sets corresponding to those Steiner nodes in the calculated Steiner network on $(H,c,r)$.

\begin{algorithm}[ht!]
\caption{\textbf{{\sc Algorithm For constrained MDSC}}} \label{algo5-11-1}
{\bf Input:} An MDSC instance $(E,\mathcal S,c,r)$

{\bf Output:} A sub-collection $\mathcal S'$.
\begin{algorithmic}[1]
\State Find an optimal (fractional) solution $(x^f,y^f,l^f)$ to linear program \eqref{eq12}.
\State Let $Y_i=\{e\in E\colon 2^{-(i+1)}<y^f_e\leq 2^{-i}\}$ for $0\leq i\leq I-1$ and $Y_I=\{e\in E\colon y^f_e\leq 2^{-I}\}$, where $I=2\lfloor\log n\rfloor-1$.
\State Let $i_0$ be an index such that $|Y_{i_0}|\geq 2^{i_0}/(I+1)$.\label{line3}
\State Find an approximation solution $H'$ to NWSN on auxiliary instance $(H,c,r)$.
\State Output $\mathcal S'=V(H')\setminus\{Y_{i_0}\cup s\}$.
\end{algorithmic}
\end{algorithm}

The rationale behind the algorithm will be manifested through the analysis in the following subsection.

\vskip 0.2cm \subsection{Theoretical Analysis}

Notice that any feasible solution to the NWSN problem on instance $(H,c,r)$ induces a feasible solution to the multi-cover problem on instance $(Y_{i_0},S,c,r)$. In fact, suppose element $e\in Y_{i_0}$ is connected to node $s$ by $r_{s,e}=r_e$ edge-disjoint paths which has the form of $\{sS_ie\}_{i=1}^{r_e}$, then $\{S_i\}_{i=1}^{r_e}$ fully covers element $e$. Taking the union of such sets will fully cover all elements in $Y_{i_0}$.

To analyze the correctness and the performance ratio, we first give an LP-relaxation for the set multi-cover problem and an LP-relaxation for the NWSN problem.

\vskip 0.2cm {\bf LP-relaxation for set multi-cover.} Similar to the construction of integer program \eqref{eq11}, the multi-cover problem on instance $(Y_{i_0},\mathcal S,c,r)$ can be formulated as an integer linear program whose relaxation is as follows:
\begin{align}
\min &\ \sum_{S\in\mathcal S}c_Sx_S\label{eq5-10-2}\\
s.t. &\ \sum_{ Q: Q\in \Omega_e} l_Q\geq 1 \ \mbox{for}\ e\in Y_{i_0}\nonumber\\
&\ x_S\geq \sum_{Q:S\in Q\in \Omega_e} l_Q \ \mbox{for}\ e\in Y_{i_0},\  S\in\mathcal S\nonumber\\
&\ x_S\geq 0\ \mbox{for}\ S\in\mathcal S\nonumber\\
&\ l_Q\geq 0\ \mbox{for}\ \mbox{for}\ Q\in \Omega_e \ \mbox{for some}\ e\in E\nonumber
\end{align}


\vskip 0.2cm{\bf LP-relaxation for NWSN.} Next, consider the node-weighted Steiner network problem. For each pair of nodes $s$ and $t$, an {\em $r_{s,t}$-path-set} is a set of $r_{s,t}$ edge-disjoint $(s,t)$-paths in $G$. Denote by $\mathcal P_{s,t}$ the family of all $r_{s,t}$-path-sets and let $\mathcal P$ be the union of all these families. The following linear program LP$_3$ is a relaxation for the NWSN problem which was presented in \cite{Chekuri}:
\begin{align}
\min &\ \sum_{v\in V}c_vx_v\label{eq5-11-2}\\
s.t. &\ \sum_{ P: P\in \mathcal P_{s,t}} l_P\geq 1 \ \mbox{for}\ s,t\in V\nonumber\\
&\ x_v\geq \sum_{P:v\in P\in \mathcal P_{s,t}} l_P \ \mbox{for}\ v\ \mbox{and}\ s,t\in V\nonumber\\
&\ x_v\geq 0\ \mbox{for}\ v\in V\nonumber\\
&\ l_P\geq 0\ \mbox{for}\ P\in \mathcal P\nonumber
\end{align}
In fact, for the corresponding integral formulation in which $l_P$ and $x_v$ can only take values from $\{0,1\}$, $l_P$ indicates whether path-set $P$ is chosen and $x_v$ indicates whether node $v$ is chosen. The model in \cite{Chekuri} uses equality instead of inequality in the first constraint, whose meaning is that for each pair of nodes $s$ and $t$, exactly one $r_{s,t}$-path-set is chosen, and thus the connectivity requirement between $s$ and $t$ is satisfied. The second constraint says that if node $v$ belongs to some chosen path-set, then $v$ must be chosen. Hence the chosen nodes are those nodes on the union of chosen path-sets, and the objective is to minimize the weight of those chosen nodes. When relaxing variables by allowing fractional values, any optimal solution automatically has $l_P\leq 1$, $x_v\leq 1$, and $\sum_{ P: P\in \mathcal P_{s,t}} l_P=1$. Hence it does not matter if we relax the first constraint to be inequality and do not explicitly require $x_S$ and $l_P$ to be no greater than 1.

Now, we are ready to analyze the performance ratio of Algorithm \ref{algo5-11-1}.

\begin{theorem}\label{thm5-11-3}
For $n\geq 32$, Algorithm \ref{algo5-11-1} has performance ratio at most $O(r_{\max}(\log n)^2)$ for constrained MDSC.
\end{theorem}
\begin{proof}
We prove the theorem step by step by first establishing the following three claims.

\vskip 0.2cm {\bf Claim 1.} An index $i_0$ as in Line 3 of Algorithm \ref{algo5-11-1} exists and $i_0\leq I-1$.

In fact, since $E$ is decomposed into $I+1$ parts $Y_0\cup Y_1\cup\cdots\cup Y_I$, by the constraint $\sum_{e\in E}y_e^f=1$, there exists an index $i_0$ such that $\sum_{e\in Y_{i_0}}y_e^f\geq 1/(I+1)$. Since $y_e^f\leq 2^{-i_0}$ for every $e\in Y_{i_0}$, we have $|Y_{i_0}|\geq 2^{i_0}/(I+1)$.

Since $I=2\lfloor \log n\rfloor -1$ and $y_e^f\leq 2^{-I}$ for each $e\in Y_I$, it can be calculated that $\sum_{e\in Y_I}y_e^f\leq n2^{-I}<1/(I+1)$ for $n\geq 32$. Hence the above $i_0\leq I-1$.

\vskip 0.2cm {\bf Claim 2.} $opt_{LP_2}\leq 2^{i_0+1}opt_{MDSC}$, where $opt_{LP_2}$ is the optimal value of linear program \eqref{eq5-10-2}.

Let $\hat x_S=2^{i_0+1}x_S^f$ for each set $S$ and let $\hat l_Q=2^{i_0+1} {l_Q}^f$ for each cover-set $Q$. For any element $e\in Y_{i_0}$ and any set $S\in\mathcal S$, we have
$$
\hat x_S=2^{i_0+1}x_S^f \geq 2^{i_0+1}\sum_{Q:S\in Q\in \Omega_e} {l_Q}^f=\sum_{Q:S\in Q\in \Omega_e}\hat l_Q.
$$
Since $i_0\leq I-1$, we have $y_e^f\geq 2^{-(i_0+1)}$ for every $e\in Y_{i_0}$. Hence
\begin{eqnarray*}
 & & \sum_{ Q: Q\in \Omega_e} \hat l_Q=2^{(i_0+1)}\sum_{ Q: Q\in \Omega_e} {l_Q}^f\geq 2^{(i_0+1)}y_e^f \geq 1\ \\
 & & \mbox{for every}\ e\in Y_{i_0}.
\end{eqnarray*}
This implies that $\{\hat x_S, \hat l_Q\}$ is a feasible solution to \eqref{eq5-10-2}. Hence
\begin{eqnarray*}
opt_{LP_2} \leq \sum_{S\in\mathcal S}c_S\hat x_S=2^{i_0+1}\sum_{S\in\mathcal S}c_Sx_S^f=2^{i_0+1}opt_{LP_1}\leq 2^{i_0+1}opt_{MDSC},
\end{eqnarray*}
where the last inequality comes from Lemma \ref{lem16-3-10}.

\vskip 0.2cm {\bf Claim 3.} $opt_{LP3}\leq opt_{LP_2}$, where $opt_{LP_3}$ is the optimal value of linear program \eqref{eq5-11-2}.

Suppose $(x,l)$ is a feasible solution to \eqref{eq5-10-2}. For each $r_e$-cover-set $Q$, let $P(Q)=\{sSe\}_{S\in Q}$ be the $r_{s,e}$-path set corresponding to $Q$. Setting $l_{P(Q)}=l_Q$ will induce a feasible solution to \eqref{eq5-11-2}. The claim is proved.

\vskip 0.2cm It was shown by Chekuri {\it et.al} in paper \cite{Chekuri} that the integrality gap for linear program \eqref{eq5-11-2} is $O(r_{\max}\log n)$. Combining this with Claim 2 and Claim 3, the output $\mathcal S'$ of Algorithm \ref{algo5-11-1} has cost
$$
c(\mathcal S')= 2^{i_0+1}O(r_{\max}\log n)opt_{MDSC}.
$$
Since $\mathcal S'$ fully covers all elements in $Y_{i_0}$, using the definition of $I$, we have
\begin{eqnarray*}
\frac{c(\mathcal S')}{|\mathcal C(\mathcal S')|} \leq \frac{c(\mathcal S')}{|Y_{i_0}|}\leq \frac{2^{i_0+1}O(r_{\max}\log n)opt_{MDSC}}{2^{i_0}/(I+1)}=O(r_{\max}(\log n)^2)opt_{MDSC}.
\end{eqnarray*}
The theorem is proved.
\end{proof}

Combining Theorem \ref{thm5-11-3} with Theorem \ref{th1}, we have the following result.

\begin{corollary}
PSMC has an $(O(\frac{r_{\max}\log^2n}{\varepsilon}),1-\varepsilon)$-bicriteria algorithm.
\end{corollary}

\section{Conclusion and Discussion}\label{secConclude}
In this paper, we studied the partial set multi-cover problem (PSMC). By proposing a new NP-hard problem called minimum density sub-collection (MDSC) and designing an $O(r_{\max}\log^2n)$-approximation for MDSC, we obtained an $(\frac{O(r_{\max}\log^2n}{\varepsilon}),1-\varepsilon)$-bicritera algorithm for PSMC. 

Our studies show that PSMC is a very challenging problem. One reason is that it possesses an ``all-or-nothing'' property. As an illustration, suppose the covering requirement for each element is 10. Covering an element 9 times has the same effect as not covering it at all. So, although the algorithm has strived to pick a large amount of sets, it is still possible that only very few elements have their covering requirements satisfied. Since what matters is only the number of fully covered elements, a lot of efforts might have been wasted on fruitless covering on those not fully covered elements. Thus, in order to obtain a good approximation, one has to control the wasted. Such an all-or-nothing phenomenon is interesting and appear frequently in the real world. New ideas are needed and conquering such a problem will have a great theoretical value.

Our performance ratio depends on the maximum covering requirement $r_{\max}$. New ideas have to be further explored to design algorithms without dependence on $r_{\max}$. Designing approximation algorithms without violation is another challenging problem.


\section*{Acknowledgment}
This research is supported by NSFC (11771013, 11531011, 61751303).



%

\section*{Appendix: A Flaw in the Conference Version.}
A preliminary version of this paper was presented in INFOCOM2017. Making using of an $\alpha$ approximation algorithm for MDSC, it was claimed that one can obtain an $\alpha H\lceil qn\rceil$-approximation algorithm for PSMC. However, there is a flaw. The algorithm in that paper greedily selects densest sub-collections until at least $qn$ elements are fully covered. Then it prunes the last sub-collection $\mathcal R$ by greedily selecting sub-collections of $\mathcal R$ consisting of at most $r_{\max}$ sets until the covering requirement is satisfied. Suppose the sub-collections obtained in the pruning step are $\mathcal R_1',\ldots,\mathcal R_l'$. The approximation analysis relies on the following inequality:
\begin{equation}\label{eq180203-1}
\frac{c(\mathcal R_1')}{|\mathcal C(\mathcal R_1')|}\leq \frac{c(\mathcal R_2')}{|\mathcal C(\mathcal R_2')|}\leq \cdots \leq \frac{c(\mathcal R_l')}{|\mathcal C(\mathcal R_l')|}.
\end{equation}
However, this is not true. Consider the following example
\begin{example}
$\mathcal S=\{S_1,S_2,S_3\}$ with $S_1=\{e_1,e_2\}$, $S_2=\{e_1,e_3\}$, $S_3=\{e_2,e_3\}$, $r(e_1)=r(e_2)=r(e_3)=r_{\max}=2$, $c(S_1)=c(S_2)=c(S_3)=1$, and $q=2/3$.
\end{example}
For this example, the densest sub-collection of $\mathcal S$ is $\mathcal R=\mathcal S$. Then the pruning step selects $\mathcal R_1'=\{S_1,S_2\}$ and $\mathcal R_2'=\{S_3\}$ sequentially. Notice that
$$\frac{c(\mathcal R_1')}{|\mathcal C(\mathcal R_1')|}=2>\frac{1}{2}=\frac{c(\mathcal R_2')}{|\mathcal C(\mathcal R_2')|}.$$

The reason why inequality \eqref{eq180203-1} does not hold is because $|\mathcal C(\mathcal R')|$ is not a submodular function, and it is difficult to bypass this obstacle. Obtaining an approximation algorithm achieving a guaranteed performance ratio in the classic sense is a very challenging problem.

\end{document}